\documentclass[nonacm]{acmart}

\AtBeginDocument{%
  }

\usepackage{enumitem}
\usepackage{amsthm}
\numberwithin{equation}{section}
\usepackage[ruled]{algorithm2e}
\usepackage{tcolorbox}

\newtheorem{theorem}{Theorem}[section]
\newtheorem{corollary}{Corollary}[theorem]

\newtheorem{prop}[theorem]{Proposition}

\newcommand{\ket}[1]{\vert #1 \rangle}

\newcommand{\frameworkname}{QCCA}

\def \frameworkname {QCCA}

\author{Anbang Wu, Jingwen Leng, Minyi Guo\\
Shanghai Jiao Tong University	
}

\begin{document}

\begin{abstract}
Existing quantum compiling efforts typically rely on a limited set of two-qubit gates, such as CX, to implement general quantum operations. However, with advancements in quantum hardware, an increasing variety of two-qubit gates are being natively implemented and calibrated on devices, including Fsim gates on Google’s Sycamore processor and partially entangling MS gates on IonQ hardware. Studies have shown that leveraging such native two-qubit gates can enhance the performance of quantum computations. While dedicated compiler optimizations can be developed to exploit the potential of such gates, those efforts may soon become obsolete as more native gates emerge. This highlights the need for a unified framework to efficiently integrate arbitrary two-qubit gates into quantum circuit compilation.
To address this challenge, we propose a compiler framework that enables analytical and computationally-cheap conversion between any two two-qubit operations, and moreover, low-overhead quantum circuits despite varied gate choices. This framework is rooted in the KAK decomposition and an analysis of the Cartan coordinates of two-qubit operations, which reveal the homotopy underlying two-qubit quantum transformations. 
With the proposed framework, any general two-qubit operation, irrespective of its specific matrix details, can serve as the intermediate representation (IR) for circuit compilation. This IR can then be near-optimally transpiled into various native gates using our approach.
We evaluate the effectiveness of the proposed framework in quantum instruction set design.
Compared to the state-of-the-art method, which relies on numerical search, our approach reduces the time overhead of compilation with diverse instruction sets by up to $2.5\times 10^5$ times in tested cases.
\end{abstract}
\title{Design the Quantum Instruction Set with the Cartan Coordinate Analysis Framework}
\maketitle


\section{Introduction}

With the development of quantum hardware, more native entangling gates beyond the well-known CX gate, emerge on mainstream quantum processors, e.g., the fSim gates on Google’s Sycamore processor \cite{Arute2019QuantumSU}, partially entangling MS gates on IonQ hardware \cite{Chen2024benchmarkingtrapped}, $\sqrt[n]{\text{iSWAP}}$ family in SNAIL-based quantum module \cite{McKinney2022CoDesignedAF}.
Existing works \cite{lao2021designing,lin2022let,peterson2022optimal, Chen2023OneGS} demonstrated that using those native gates can improve the fidelity or reduce the latency of quantum circuits. Effective calibration schemes \cite{lao2021designing,lin2022let, Chen2023OneGS} for those native two-qubit gates emerge at the same time, further supporting their usage in practical quantum computing. 

While we can develop specialized compiler optimizations to efficiently exploit such gates in quantum circuits, those efforts may soon become obsolete as novel native gates emerge, considering the variety of qubit technologies.
It is thus natural to ask whether it is possible to come up with a unified framework that can efficiently incorporate arbitrary native gates into quantum circuit compilation and execution. 

Unfortunately, existing efforts\cite{lao2021designing,lin2022let} heavily rely on numerical methods to compile circuits with an extended gate set, leading to tremendous computation cost.
These numerical compilers will first construct a candidate circuit consisting of interested two-qubit gates and adjustable single-qubit gates, and then tune the parameter of single-qubit gates so that the candidate circuit approximates the target one, in terms of matrix distance. More two-qubit gates and tunable single-qubit gates will be added to the candidate circuit if the approximation error is more than desired. Such a numerical search process will become slow if there are many tunable single-qubit gates in the candidate circuit. For example, with the numerical compiler BQSKit \cite{lao2021designing,BQSKIT}, decomposing SWAP gate with the XX interaction gate $XX(\frac{\pi}{28})$ will take about two minutes, let alone to employ this gate in quantum circuits consist of several hundred of varied two-qubit operators. Peterson et al.\cite{peterson2022optimal} propose to find the closest candidate circuit to the target per search step analytically, largely reducing the compilation time. However, their work still does not overcome the path-like nature of the compilation process, and will inevitably incur higher compilation time if more search steps are needed for the target program. Moreover, their work only focuses on XX interaction gates, and cannot handle with arbitrary two-qubit native gates.

One key reason for the expense computational cost of existing works is that they do not fully utilize the information of the native gate set and the target circuit to be compiled. According to the KAK decomposition, the effect of any two-qubit gate can be represented by the rotation over XX, YY, and ZZ axes. Then the conversion/compilation of quantum operations into the native two-qubit gate can be transformed into an angle manipulation problem, which is easier to solve intuitively. For example, three continuous application of the two-qubit gate with a rotation angle $\theta$ should be able to fulfill a two-qubit operation with a rotation angle less than $3\theta$. While this intuition is not always correct, it indicates that such angle information (aka Cartan coordinate) of gates can simplify the quantum circuit compilation.

Based on this insight, we propose a computationally-efficient compiling framework, named \frameworkname,  which enables analytical conversion between any two two-qubit entangling gates $g_1$ and $g_2$ and moreover low-overhead quantum circuits regardless of instruction set selection.
The compilation is based on the idea that a tunable single-qubit rotation gate $g^{1Q}(\theta)$, when `dressed' by two-qubit entangling gates $g^{2Q}$, can become a tunable two-qubit-level rotation gate (e.g., $g^{2Q} g^{1Q}(\theta) g^{2Q}$), allowing continuous deforming of two-qubit gates, when viewing them as rotation angles (by KAK decomposition). Then to compile $g_2$ into $g_1$, we only need \textit{a few constant steps} to determine a global $\theta$ that induces the required deforming from $g_1$ to $g_2$, avoiding local path-finding process in previous works. The resulting compilation time is small (far less than 1ms) and will not be affected by the pair of $g_1$ and $g_2$. 
Also, the resulting gate cost by our approach differs the optimal cost by at most additive constants, which is small and often disappears.
With the proposed framework, we can near-optimally  translate any quantum circuits from one instruction set (e.g., CX+U3 \cite{qiskit2024}) to another emerging instruction set, maintaining the low overhead of quantum circuits.  This empowers using any general two-qubit operation, regardless of its specific matrix details,  as the intermediate representation (IR) for circuit compilation since such IR can now be cheaply  transpiled into various native gates with our approach.

Overall, this work proposes a unified and computational-efficient framework which enables using any two-qubit gates as the basis of quantum circuit compilation. It can help us quickly deploy quantum circuits to a novel quantum platform and fully exploit the potential of its native gates. For a quantum platform which support diverse native gates, our framework can also be used to quickly explore the instruction set design space, rapidly evaluating the effectiveness of different instruction sets in a quantum program context, which is otherwise not possible with numerical compilers. Also, allowing a flexible instruction set in compilation can help suppress errors in quantum circuits. For example, the CX gate's behavior may be distorted by quantum noises. Our framework can directly include this noise-induced CX variant in compilation without incurring heavy computation cost, avoiding errors caused by distorted gate behaviors.

Experimental results demonstrate the effectiveness of the proposed framework in quantum instruction set design.
Compared to the state-of-the-art method, our approach speedup the circuit compilation over diverse instruction sets by up to $2.5\times 10^5$ times, among various test cases. Based on the proposed compiler, we can compute and compare the effectiveness of thousands of different instruction set configurations within seconds, significantly accelerating the quantum instruction set design. 

\section{Background}\label{sect:bg}

In this section, we introduce the basic background knowledge of quantum computation and quantum compilation. For more details, we recommend \cite{nielsen2010quantum, khaneja2000cartan, shende2005synthesis, shende2003minimal, dawson2005solovay, barenco1995elementary, Bergholm2004QuantumCW, de2016block, Vatan2003OptimalQC, Vidal2003UniversalQC, Aho2003CompilingQC, Cybenko2001ReducingQC, Vartiainen2003EfficientDO, Mttnen2004QuantumCF}.

\subsection{Basics of Quantum Computation} 

\textbf{Qubit and quantum state:} The basic information unit of quantum computation is called {\it qubit}. The state of a qubit is a unit vector $\ket{\phi} = \alpha \ket{0} + \beta \ket{1}, \vert \alpha\vert^2+\vert \beta\vert^2 = 1$ in the complex Hilbert space spanned by basis quantum states $\{\ket{0}, \ket{1}\}$. 

\textbf{Quantum gates:} The state of a qubit can be manipulated by single-qubit/1Q quantum gates/operators, e.g., Pauli operators
\begin{equation*}
	Z = \begin{pmatrix}
		1 & 0\\
		0 & -1
	\end{pmatrix}, X = \begin{pmatrix}
		0 & 1\\
		1 & 0
	\end{pmatrix}, Y = \begin{pmatrix}
		0 & -i\\
		i & 0
	\end{pmatrix},
	I = \begin{pmatrix}
		1 & 0\\
		0 & 1
	\end{pmatrix}.
\end{equation*}
Other commonly used gates include the phase gate $S$, the Hadamard gate $H$ and the CX gate
$$ S = \begin{pmatrix}
	1 & 0 \\ 0 & i
\end{pmatrix}, \quad H = \frac{1}{\sqrt{2}}\begin{pmatrix}
	1 & 1 \\ 1 & -1
\end{pmatrix},\quad CX = \begin{pmatrix}
I & \\
& X
\end{pmatrix}.$$ 
Here the CX gate is a two-qubit/2Q gate and can manipulate the quantum state of two qubits simultaneously.

Based on Pauli operators, we can define 1Q rotation gates 
$$Z(\beta) = e^{i\beta Z} = \cos\beta I + i\sin\theta Z,\quad X(\beta) = e^{i\beta X},\ Y(\beta) = e^{i\beta Y}.$$

Likewise, we can define 2Q rotation gates
$$ZZ(\beta) = e^{i\beta Z\otimes Z}, XX(\beta) = e^{i\beta X\otimes X}, YY(\beta) = e^{i\beta Y\otimes Y}.$$

For more details about quantum gates, please refer to \cite{nielsen2010quantum}. 

\textbf{To differentiate gates/operators applied to different qubits, this work uses subscripts}, e.g., $Z_0$ means to apply the Z gate on qubit $q_0$, $H_{0,1}$ means to apply H gates on both $q_0$ and $q_1$.

\subsection{KAK Decomposition and Local Equivalence}

\textbf{Cartan coordinate}: KAK decomposition \cite{khaneja2000cartan} can transform any 2Q special unitary matrix $U \in SU(4)$ into the following canonical form:
\begin{equation}
	U = k_g(a\otimes b) e^{i (\eta_x X_0X_1 + \eta_y Y_0Y_1 + \eta_z Z_0Z_1) } (c\otimes d),
\end{equation}
where the global phase $k_g \in \{1, i\}$, $a, b, c, d \in SU(2)$ (i.e., 1Q gates), $0 \le \vert \eta_z \vert \le \eta_y \le \eta_x \le \frac{\pi}{4} $. If $\eta_x = \frac{\pi}{4}$, $\eta_z \ge 0$. $(\eta_z,\eta_y,\eta_x)$ is called  \textit{Cartan coordinate} of $U$. 

\textbf{Cartan coordinate manipulation}: A straightforward interpretation of $\eta_x \le \frac{\pi}{4}$ is that, if $\eta_x > \frac{\pi}{4}$, we can simply multiply $-iX_0X_1=e^{-i\frac{\pi}{2}X_0X_1}$ to $e^{i\eta_x X_0X_1}$, then the norm of $\eta_x' = \eta_x - \frac{\pi}{2}$ is less than $\frac{\pi}{4}$. As for the negative sign of $\eta_x'$, we can simply cancel it out by dressing XX rotation with Y rotations, e.g., $Y_0e^{i\theta_x' XX}Y_0^\dagger = e^{-i\theta_x' XX}$. Moreover, we can even transform a XX rotation angle into a YY/ZZ rotation one, e.g., $(S\otimes S)e^{i\theta_x XX}(S^\dagger\otimes S^\dagger) = e^{i\theta_x YY}$, $(H\otimes H)e^{i\theta_x XX}(H\otimes H) = e^{i\theta_x ZZ}$.
Those discussions actually present us a tool to manipulate the Cartan coordinate.

\textbf{Local equivalence}: Two-qubit unitary matrices $U_1$ and $U_2$ are called locally equivalent if they share the same Cartan coordinate.
Based on the notion, any two-qubit operations, ignoring the global phase, can be categorized into the following three templates:
\begin{align}
	D_a(\theta_x) &= e^{i \theta_x X_0X_1}, \\
	D_b(\theta_x, \theta_y) &= e^{i (\theta_x X_0X_1 + \theta_y Y_0Y_1)} \\
	D_c(\theta_x, \theta_y, \theta_z) &= e^{i (\theta_x X_0X_1 + \theta_y Y_0Y_1 + \theta_z Z_0Z_1)}
\end{align}

Without loss of generality, we assume $\theta_x, \theta_y \ge 0$.
The template $D_a$ is widely available in various quantum computing platforms, including IBM and Rigetti superconducting devices~\cite{qiskit2024, smith2016practical}. Chen et al.~\cite{Chen2023OneGS} recently proposes a pulse scheme to implement the template $D_b$ on superconducting devices. Google's fSim gate~\cite{Arute2019QuantumSU} and IonQ's partially entangling MS gate are locally equivalent to $D_b$. 

\textbf{Compilation overhead/cost}: The goal of quantum compiling is to turn a quantum circuit/program into a hardware executable one, i.e., synthesizing operations in the circuit into a gate set native to hardware, such as the gate set consisting of the CX gate and all 1Q gates. The native gate set is also called a quantum instruction set. Gates in the  set are termed \textit{basis gates}. For a two-qubit basis gate $g$, and a quantum program $P$, if we need $n$ invocations of $g$ to synthesize all quantum operations in $P$, we say the \textit{decomposition cost} of $P$ with $g$ is $n$.

In this work, we consider compilation with any 2Q gate $D_a, D_b, D_c$, to enable a flexible instruction set in quantum computing.
\section{Challenge and Motivation}

As discussed in the introduction, having a unified compiler framework over diverse instruction sets is critical for unleashing the power of quantum computing.
It can help us quickly deploy quantum circuits to a novel quantum platform and fully exploit the potential of its native gates. For example, the QAOA benchmark consists of ZZ rotation gates of varied parameters. To decompose a ZZ rotation gate, e.g., $ZZ(\frac{\pi}{8})$ with CX, we will need at least two invocations of CX gates 
\begin{equation}\label{equ:zzdecomp}
	ZZ_{0,1}(\frac{\pi}{8}) = CX_{0,1}\ Z_1(\frac{\pi}{8})\ CX_{0,1}.
\end{equation}

Instead, if the instruction set includes $ZZ(\frac{\pi}{8})$, the resulting circuit may have a smaller latency, since from the perspective of pulse engineering, the latency of $ZZ(\frac{\pi}{8})$ is less than a half of that of CX \cite{peterson2022optimal}, let alone two CX are needed in Equation \ref{equ:zzdecomp}. 

Despite this potential, it still remains unclear how to efficiently do compilation with arbitrary two-qubit gates. The numerical compiler \cite{BQSKIT} may exhibit a heavy compilation cost. Assume we need $n$ invocations of $g_1$ to compile the operation $g_2$, numerical compilers will take at least $n$ search steps and $36n^2$ time. Ideally, we need at least $O(n)$ compilation time which is no more complex than sequentially appending $n$ occurrence of $g_1$ to the circuit.

\begin{center}
	\textit{The challenge is that, whether is it possible to reduce the compilation complexity to constant computation steps and time?}
\end{center}

\begin{figure}[t]
	\centering
	\includegraphics{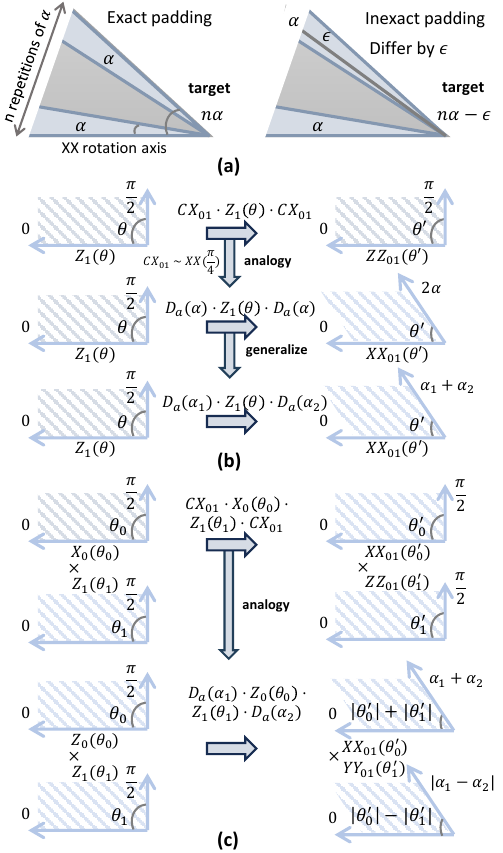}
	\caption{A motivation example of compilation with $D_a$. (a) Intuitive compilation by rotation angle addition. (b) Compiling residual rotation $XX(\epsilon)$ with $D_a$ by mimicking CX-based synthesis. Note that $2\alpha$ or $\alpha_1+\alpha_2$ larger than $\epsilon$. (c) Converting $D_a$ from only XX rotation to XX+YY hybrid rotations, to handle diverse residual rotations (axes).}\label{fig:mot}
\end{figure}

\textbf{Our insight} to tackle this challenge is that, we can use KAK decomposition to project two-qubit operations into rotation angles, and turn compilation into the more intuitive angle manipulation. To be concrete, let us consider the relatively simple gate synthesis with $D_a$. assuming having $D_a(\alpha)$ in the instruction set, let us consider the example to compile a gate $g_1$ locally equivalent to $XX(n\alpha)$.

\textbf{Rotation padding}: As shown in the left panel of Figure \ref{fig:mot}(a), we only need to repeat $D_a(\alpha)$ n times and plus 1Q gates in the KAK decomposition of $g_1$ to finish the compilation. In this process, the only computation is induced by $n = \lceil n\alpha/\alpha \rceil$. 

Unfortunately, as shown in Figure \ref{fig:mot}(a) right panel, the rotation padding may not always be exact and we may have a residual rotation $XX(\epsilon)$, $\epsilon < \alpha$, ignoring the negative sign of $\epsilon$ which can be easily removed by the Cartan coordinate manipulation discussed in Section \ref{sect:bg}.

\textbf{Residual synthesizing}: We observe that, the CX gate is locally equivalent to $XX(\frac{\pi}{4})$ but is able to produce any tunable 2Q ZZ interaction between $[0, \frac{\pi}{2}]$, when accompanied with tunable single-qubit rotation gate, as shown in Equation \ref{equ:zzdecomp} and Figure \ref{fig:mot}(b). It is thus natural to deduce that, by replacing $CX$ in Equation \ref{equ:zzdecomp} with $D_a$, we can use the combination $D_a(\alpha)\ Z_1(\theta)\ D_a(\alpha)$ to implement the residual rotation $XX(\epsilon)$. This analogy is indeed true and will be elaborated in later sections. 
As shown in Figure \ref{fig:mot}(b), in this compilation process, the only computation is resulted by computing the $\theta$ which leads to $XX(\epsilon)$. We may reuse the numerical compiler to find $\theta$ and in this case one search step is sufficient. In our work, we further reduce the computation cost by instead proposing analytical solutions of $\theta$. 
Overall, the compilation time for residual rotation is independent of $g_1$ and $D_a$ and is a constant factor.

\textbf{Hybrid synthesizing}: Now, let us consider a slightly more complicated example, which is to compile $g_2$ locally equivalent to $XX(n_1\alpha_1-\epsilon_1)\cdot YY(n_2\alpha_2-\epsilon_2)$. As discussed in Section \ref{sect:bg}, we can convert the XX rotation $D_a$ to YY/ZZ rotations. Therefore, after the rotation padding step, we need to compile a hybrid residual XX+YY rotation  gate. Intuitively, we can tackle this challenge as in Figure \ref{fig:mot}(c): we can produce hybrid rotation from $D_a$ by mimicking the CX-based synthesis. This technique together with the previous two is sufficient to enable universal compilation with $D_a$. Likewise, the computation time for this procedure is also a constant factor. 

Overall, from the simple example of gate synthesis by $D_a$ (see Figure \ref{fig:mot}), we can envision a promising strategy to enable constant computation cost for compilation with diverse 2Q gates: first applying rotation padding to reduce the scale of operations to be compiled, and then synthesizing the remaining small-scale residual rotations. Based on this idea, we further propose efficient compilation with $D_b$ and $D_c$. In those cases, we cannot simply mimicking the behavior of CX-based synthesis. Fortunately, we discover that, from the perspective of Lie algebra, we can  generalize techniques in Figure \ref{fig:mot} to work with $D_b$ and $D_c$. The flexibility of the proposed strategy, i.e., rotation padding + residual synthesizing, even enables us to do hybrid compilation with $D_a, D_b, D_c$ simultaneously.

Moreover, we demonstrate the optimality of the proposed strategy by establishing a \textbf{triangle inequality} on Cartan coordinates. To some extent, we show that the entangling power of a two-qubit operation is proportional to the L1 norm of its Cartan coordinate, making its decomposition cost no less than the estimation by rotation padding. This indicates that, our compiler differs from the optimal one by at most an additive factor induced by residual rotation compilation, as will be elaborated in later sections.

Comparing to the state-of-the-art, the proposed strategy avoids the heavy computation cost of numerical compilers \cite{BQSKIT, lao2021designing, lin2022let}.
Peterson et al. \cite{peterson2020fixed} proposed a similar formulae to the one shown in Figure \ref{fig:mot}(c). However, their work does not consider rotation padding, leading to a path-finding-based expense compilation, similar to numerical compilers. Moreover, their work does not support the compilation with $D_b$ and $D_c$.
\section{The \frameworkname{} Compiler Framework}\label{sec:theory}

In this section, we first analyze Cartan coordinate transformations, establish analytical relations beyond CX-based ones that will turn a tunable 1Q gate to a tunable 2Q gate. We then
extend techniques in Figure \ref{fig:mot} to enable compilation with all $D_a$, $D_b$ and $D_c$. Finally, we provides optimality analysis of the proposed compiler framework.

\subsection{Cartan Coordinate Analysis}

Firstly, by mimicking the analytical construction in Equation \ref{equ:zzdecomp}, 
\begin{prop}\label{prop:datransform} For $\theta_x, \theta_x' > 0$, we can find $\eta, \tau$, s.t.
	\begin{align} \label{equ:daz}
		e^{i\eta X_0X_1} &= Z_1(\tau) D_a(\theta_x) Z_1(\beta-\pi/2) D_a(\theta_x') Z_1(\tau), \\
		&= Z_1(\tau) D_a(\theta_x) Z_1(\beta) D_a^\dagger(\theta_x') Z_1(\tau-\pi/2),
	\end{align}
	where $\sin\eta = \sin\beta\sin \theta_x+\theta_x'$, $\tan 2\tau = \frac{\cos\beta}{\sin\beta\cos \theta_x+\theta_x'}$.
\end{prop}

Different from Equation \ref{equ:zzdecomp}, the construction with $D_a$ requires additional 1Q gates $Z_1(\tau)$, called \textbf{calibration operators}.
To extend such construction to $D_b$, we need a new perspective to interpret Equation \ref{equ:daz}, allowing an appropriate way to replace $D_a$ with $D_b$.

Indeed, from the perspective of Lie algebra theory, the right hand side (RHS) term of Equation~\ref{equ:daz} is in the space $exp(V)$, where $V$ is the Cartan algebra $i\{X_0X_1,X_0Y_1, I_0Z_1\}$. Then according to the Cartan decomposition theorem~\cite{khaneja2000cartan}, 
$$exp(V) = exp(i\{I_0Z_1\}) exp(i\{X_0X_1\}) exp(i\{I_0Z_1\}).$$
This formula is just the restatement of Equation~\ref{equ:daz} since $exp(i\{IZ\})$ is exactly the Z-rotation gate on $q_1$.  Replacing $D_a$ with $D_b$, we then have $V = i\{X_0X_1, Y_0Y_1, X_0Y_1, Y_0X_1, I_0Z_1, Z_0I_1\}$ and

\begin{prop}\label{prop:dbtransform} For $\theta_x, \theta_y, \theta_x', \theta_y'$, $\theta_x \ge \theta_y > 0$, $\theta_x' \ge \theta_y' > 0$,
	\begin{align} \label{equ:dby}
		\begin{aligned}
			e^{i(\eta_x X_0X_1 + \eta_y Y_0Y_1)} = Z_0(\tau_0) Z_1(\tau_1) D_b(\theta_x,\theta_y)\cdot \\Z_0(\beta_0) Z_1(\beta_1)\cdot \\ D_b(\theta_x',\theta_y')  Z_0(\tau_2) Z_1(\tau_3),
		\end{aligned}
	\end{align}
	where $\eta_x, \eta_y$ satisfy
	\begin{align}
		\begin{aligned}
			\cos^2(\eta_x \pm \eta_y) = \cos^2(\beta_0\mp\beta_1)\cos^2(\theta_x\pm\theta_y+\theta_x'\pm\theta_y') \\+ \sin^2(\beta_0\mp\beta_1)\cos^2(\theta_x\pm\theta_y-\theta_x'\mp\theta_y'). 
		\end{aligned}
	\end{align}
	$\tau_0, \tau_1, \tau_2, \tau_3$ can be analytically determined by comparing the matrix element of LHS and RHS. Let $a_0=\tau_0+\tau_2, a_1 = \tau_1 + \tau_3, a_2 = \tau_0-\tau_2, a_3 = \tau_1-\tau_3, b_0 = \theta_x + \theta_x', b_1 = \theta_y+\theta_y', b_2 = \theta_x - \theta_x', b_3 = \theta_y-\theta_y'$, then
	\begin{align}
		\cos(\eta_x-\eta_y)\cos(a_0 + a_1) &= \cos(\beta_0+\beta_1) \cos(b_0 - b_1) \\
		\cos(\eta_x+\eta_y)\cos(a_0 - a_1) &= \cos(\beta_0-\beta_1) \cos(b_0 + b_1) \\
		\cos(\eta_x-\eta_y)\sin(a_0 + a_1) &= \sin(\beta_0+\beta_1) \cos(b_2 - b_3) \\
		\cos(\eta_x+\eta_y)\sin(a_0 - a_1) &= \sin(\beta_0-\beta_1) \cos(b_2 + b_3) \\
		\sin (\eta_x + \eta_y)\cos (a_2-a_3) &= \cos(\beta_0 - \beta_1)\sin (b_0+b_1) \\
		\sin(\eta_x-\eta_y)\cos (a_2+a_3) &= \cos(\beta_0 + \beta_1)\sin(b_0 - b_1) \\
		\sin(\eta_x-\eta_y)\sin(a_2+a_3) &= -\sin(\beta_0+\beta_1)\sin(b_2-b_3) \\
		\sin(\eta_x+\eta_y)\sin(a_2-a_3) &= - \sin (\beta_0 - \beta_1)\sin(b_2+b_3) 
	\end{align}
\end{prop}

Proposition~\ref{prop:datransform} and \ref{prop:dbtransform} reveal how to transform a 2q gate with a fixed Cartan coordinate into a 2q gate with parameterized Cartan coordinate. Those two propositions can be used to generate formulas for other transformations based on $D_a, D_b, D_c$. For example, we can also consider the Cartan algebra $V = i\{XX, YY, XI, IY, ZY, XZ\}$, which leads to another construction with $D_b$, which is equivalent to apply Proposition~\ref{prop:datransform} twice:
$$e^{i(\eta_x X_0X_1 + \eta_y Y_0Y_1)} = (e^{i\eta_x X_0X_1})(e^{i\eta_y Y_0Y_1}) = (c_1^Y D_b Y_1(\beta_x) D_b^\dagger c_1^Y) (c_0^X D_b X_0(\beta_y) D_b^\dagger c_0^X)
= c_1^Yc_0^X D_b Y_1(\beta_x) X_0(\beta_y) D_b^\dagger c_1^Yc_0^X$$
where $ c_0^Y, c_1^Y,  c_0^X, c_1^X$ are calibration operators determined by applying Proposition~\ref{prop:datransform}.

Further, Figure \ref{fig:mot}(c) is a special case of Proposition \ref{prop:dbtransform} where $\theta_y=0, \theta_y' = 0$. Also, the LHS of Equation \ref{equ:daz} and \ref{equ:dby} can be transformed into other rotation axes, as discussed in Section \ref{sect:bg}. 

\begin{corollary}\label{cor:rules} We observe the following rules for parameterized Cartan coordinate transformation:
	\begin{enumerate}
		\item (Pauli conversion) By multiplying $H_1 S_1 H_1 H_0 S_0 H_0$ and its inverse to the left and right side of the LHS and RHS term respectively in Equation~\ref{equ:daz}, we see $D_a Y(\beta) D_a^\dagger$ is locally equivalent to $e^{i\eta X_0X_1}$ up to some local Y-rotation gates.
		\item (Reduce) $D_a Z(\beta) D_a = D_a Z(\beta-\frac{\pi}{2}) D_a^\dagger Z(\frac{\pi}{2})$, $D_a Y(\beta) D_a = D_a Y(\beta-\frac{\pi}{2}) D_a^\dagger Y(\frac{\pi}{2})$, $D_b Z(\beta) D_b = D_b Z(\beta-\frac{\pi}{2}) D_b^\dagger Z(\frac{\pi}{2})$;
		\item (Downgrade I) $D_b Y(\theta) D_b^\dagger = D_a Y(\theta) D_a^\dagger$ (assuming the same $\theta_x$), $D_c Z(\theta) D_c^\dagger = D_b Z(\theta) D_b^\dagger$;
		\item (Downgrade II) $D_b Y(\theta) D_b = e^{i2\theta_y Y_0Y_1} D_a Y(\theta) D_a$, and then Rule 2;
		If $D_c^\dagger$ is not available, $D_c Z(\beta) D_c = e^{i2\theta_z Z_0Z_1} D_b Z(\theta) D_b$, and then reuse Rule 2. 
	\end{enumerate}
\end{corollary}
Note that, since $e^{i2\theta_y Y_0Y_1}$ is commutable with Y-rotation gates, $e^{i2\theta_y Y_0Y_1} D_a Y(\theta) D_a$ is locally equivalent to $e^{i2\theta_y Y_0Y_1}e^{i\eta X_0X_1}$ by Proposition~\ref{prop:datransform}. The case for $e^{i2\theta_z Z_0Z_1} D_b Z(\theta) D_b$ is similar.

\subsection{Compilation with Arbitrary 2Q Gate}

Theoretically, there is unlimited ways to apply Proposition \ref{prop:datransform}, \ref{prop:dbtransform} in circuit compilation process. In our compiler design, we choose to utilize those constructions \textbf{globally}, that is, once the operation to be compiled is known, our compiler directly produces the compiled circuit without using the path searching in state-of-the-art works~\cite{lao2021designing,peterson2022optimal,lin2022let}.

The basic idea follows Figure \ref{fig:mot}.

\subsubsection{Optimized Compilation with $D_a(\theta_x)$}\label{sect:da}
To compile a general $u$ that is locally equivalent to $XX(\eta_x)YY(\eta_y)ZZ(\eta_z)$, we need $m = \lfloor \frac{\eta_x}{\theta_x} \rfloor + \lfloor \frac{\eta_y}{\theta_x} \rfloor + \lfloor \frac{\eta_z}{\theta_x} \rfloor$ invocations of $D_a$ for \textbf{rotation padding}. 

Let the padded rotation $u_p = XX(\lfloor\frac{\eta_x}{\theta_x} \rfloor*\theta_x)\cdot YY(\lfloor\frac{\eta_y}{\theta_x} \rfloor*\theta_x)\cdot ZZ(\lfloor\frac{\eta_z}{\theta_x} \rfloor*\theta_x)$, we then consider the compilation of the \textbf{residual rotation} $u_r = XX(\gamma_x)YY(\gamma_y)ZZ(\gamma_z)$ that satisfies $\max(\vert\gamma_x \vert, \vert\gamma_y\vert, \vert\gamma_z\vert) < \theta_x$. We need at most three $D_a$ for this target. Without loss of generality, we assume $ \gamma_x \ge \gamma_y \ge \vert\gamma_z \vert \ge 0$. Then, after the first application of Proposition \ref{prop:dbtransform} by $D_a Z_0(\beta_0)Z_1(\beta_1)D_a$ (the leading $D_a$ is from $u_p$), the residual rotation will now become $(0, \gamma_y-(\theta_x-\gamma_x), \gamma_z)$. Then, one more application of Proposition \ref{prop:dbtransform} by  $H_{0,1}D_aH_{0,1} X_0(\beta_0)X_1(\beta_1)H_{0,1}D_aH_{0,1}$, with the leading $H_{0,1}D_aH_{0,1}$  from $u_p$, is enough for compilation. Note here we convert $D_a$ from XX rotation to ZZ rotation by $H_{0,1}$.

Finally, we obtain the compilation of $XX(\eta_x)YY(\eta_y)ZZ(\eta_z)$. To reach $u$, we can use KAK decomposition to find the 1Q gates that differ $u$ from $XX(\eta_x)YY(\eta_y)ZZ(\eta_z)$.

\subsubsection{Optimized Compilation with $D_b(\theta_x, \theta_y)$}\label{sect:db}
The basic compiling flow with $D_b$ follows the one with $D_a$. For the compilation target $u$ locally equivalent to $XX(\eta_x)YY(\eta_y)ZZ(\eta_z)$, we now need two steps for rotation padding. As an example,  assuming $\lfloor\frac{\eta_x}{\theta_x} \rfloor > \lfloor\frac{\eta_y}{\theta_y} \rfloor$, then in the first step, we need $\min(\lfloor\frac{\eta_x}{\theta_x} \rfloor, \lfloor\frac{\eta_y}{\theta_y} \rfloor)$ $D_b$ for padding. In the second step, if we assume $\eta_x - \lfloor\frac{\eta_x}{\theta_x}\rfloor*\theta < \eta_z$, let $\eta_x' = \eta_x - \lfloor\frac{\eta_x}{\theta_x}\rfloor\theta_x$ then we need  $\min(\lfloor\frac{\eta_z}{\theta_x} \rfloor, \lfloor\frac{\eta_x'}{\theta_y} \rfloor)$ $D_b$ for padding. Note in this step we will convert $D_b$ from XX+YY rotation to ZZ+XX rotation. Other cases in rotation padding can be discussed similarly.

As for residual rotation $u_r = XX(\gamma_x)YY(\gamma_y)ZZ(\gamma_z)$, due to the limitation of rotation padding, we may not keep $\max (\gamma_x, \gamma_y, \gamma_z) < \theta_x$. For simplicity, here we only discuss the case where $\gamma_x > \theta_x > \gamma_y > \gamma_z$. Like the case in $D_a$, we first do at most two application of Proposition \ref{prop:dbtransform} to eliminate $\gamma_y$ and $\gamma_z$, transforming the rotation axes of $D_b$ if needed. Then, for the remaining $\gamma_x'$ (should now be smaller than $\gamma_x$), we can eliminate it by $\lceil \frac{\gamma_x'}{\theta_x} \rceil$ applications of Proposition \ref{prop:dbtransform}.
Other cases of residual rotation synthesis can be discussed similarly. 

\subsubsection{Optimized Compilation with $D_c$}\label{sect:dc}
The rotation padding process with $D_c$ is almost the same as the one with $D_b$ and may  permute rotation axes of $D_c$ (with local Clifford gates) as needed. We will stop padding once we find the minimal rotation angle of the residual $u_r$ is negative. This is to avoid causing a rotation with a large negative angle, incurring extra decomposition cost. 

As for residual synthesis, we considers the identities
$
	D_c X_0 D_cX_0 = e^{i2\theta_x X_0X_1},
	D_c Y_0 D_cY_0 = e^{i2\theta_y Y_0Y_1},
	D_c Z_0 D_cZ_0 = e^{i2\theta_z Z_0Z_1}.
$ Those identities can convert $D_c$ to $D_a$ and then we can follow the flow of compiling with $D_a$. As for the rotation angle selection of $D_a$, considering the periodic property of n-qubit Pauli rotations, we can set the optimal rotation angle of $D_a$ as
\begin{align}\footnotesize
	\begin{aligned}
		\theta_x' = \max(\min(2\theta_x, \vert \frac{\pi}{2} - 2\theta_x \vert), \min(2\theta_y, 
		\vert \frac{\pi}{2} - 2\theta_y \vert), \min(\vert 2\theta_z \vert, \vert \frac{\pi}{2} - \vert 2\theta_z \vert \vert)).
	\end{aligned}
\end{align} 

\subsubsection{Putting Things Together}

\begin{algorithm}\label{algo:kak2q}
	\caption{Decomposition with any 2Q Gate}\label{alg:two}
	\textbf{Input:} operator to be decomposed $U_y$, basis gate $U_x$;\\
	Do KAK decomposition on $U_x$, and get $U_x = (a_x\otimes b_x) D_x (c_x\otimes d_x)$; \\
	Do KAK decomposition on $U_y$, and get $U_y = (a_y\otimes b_y) D_y (c_y\otimes d_y)$; \\
	\uIf{$D_x$ is of the $D_a$ type}
	{
		$C$ = Compiling $D_y$ as in Section~\ref{sect:da};
	}
	\uElseIf{$D_x$ is of the $D_b$ type}
	{
		$C$ = Compiling $D_y$ as in Section~\ref{sect:db};
	}
	\uElseIf{$D_x$ is of the $D_c$ type}
	{
		$C$ = Compiling $D_y$ as in Section~\ref{sect:dc};
	}
	$C'$ = Replacing $D_x$ in $C$ with $(a^\dagger_x\otimes b^\dagger_x) U_x (c^\dagger_x\otimes d^\dagger_x)$; \\
	\textbf{Output:} $U_y = (a_y\otimes b_y) C' (c_y\otimes d_y)$;
\end{algorithm}

Now, we can compile any n-qubit quantum unitary gate with $D_a, D_b$ or $D_c$. We can first decompose the n-qubit quantum unitary gate into a group of 2q and 1q gates and then use Algorithm~\ref{alg:two} to compile any 2q unitary gates. It is easy to see that our strategy inherently allows hybrid compilation with $D_a$, $D_b$ and $D_c$ simultaneously. The rotation padding and residual synthesis process of \frameworkname{} is applicable to all $D_a$, $D_b$ and $D_c$.

\subsection{Compilation Cost Analysis}

As for the decomposition cost,
we observe that, adding one more $D_a(\theta_x)$ to a local two-qubit circuit may not increase the L1 norm of the Cartan coordinate of the circuit by more than $\theta_x$. This conclusion is sort of like unitary-version triangle inequality, and is formalized in the following theorem:

\begin{theorem}[Lower Bound of the Decomposition Cost]\label{prop:dacost}
	Assume the count of $D_a$ in the decomposed circuit of $U \in SU(4)$ is $N$. If $N\theta_x < \frac{\pi}{4}$, then $sin (k_t(U)) \le \sin N\theta_x$, where $k_t(U)$ is the L1 norm of the Cartan coordinate of $U$.
\end{theorem}
\begin{proof}
	Please see the appendix.
\end{proof}
\begin{corollary}\label{cor:dbdc}
	Similarly, the following conclusions can be obtained by repeatedly applying Theorem~\ref{prop:dacost}:
	\begin{enumerate}
		\item For $D_b$, if $N(\theta_x+\theta_y) < \frac{\pi}{4}$ ($N$ is the decomposition cost of $U$ with $D_b$), then $\sin \, k_t(U) \le \sin N(\theta_x+\theta_y)$;
		\item For $D_c$, if $N(\theta_x+\theta_y+\vert\theta_z\vert) < \frac{\pi}{4}$ ($N$ is the decomposition cost of $U$ with $D_c$), then $\sin \, k_t(U) \le \sin N(\theta_x+\theta_y+\vert\theta_z\vert)$.
	\end{enumerate}
\end{corollary}

Theorem \ref{prop:dacost} and Corollary \ref{cor:dbdc} indicate that the L1 norm of the Cartan coordinate of quantum operations measures their entangling power. These theorems indeed provide a lower bound of decomposition cost. Even the optimal compiler cannot be better than this lower bound. As for our framework, we can see that the decomposition cost resulting from the rotation padding step approaches the lower bound. Therefore, our framework deviates from the optimal compiler only by an additive factor induced by the residual synthesis step. 

As for the compilation time, both rotation padding and residual synthesis procedure of our framework are analytical and take constant computation time, for all $D_a, D_b$ and $D_c$. 
\section{Evaluation}

In this section, we first compare the proposed compiler framework with the state-of-the-art compiler that can also work with diverse native two-qubit basis gates. We then evaluate the effectiveness of the proposed compiler in quantum instruction set design.

\subsection{Experiment Setting}

\textbf{Platform:} All experiments are performed with Python 3.8.10, on a Ubuntu 20.04 server with 16 CPU cores and 32GB RAM.

\textbf{Baseline:} We first compare the proposed compiler \frameworkname{} to the numerical search-based compiler BQSKit~\cite{BQSKIT, lao2021designing, lin2022let}, which also supports compilation with diverse two-qubit gates. 

We then compare \frameworkname{} to XXDecomposer~\cite{peterson2022optimal} in Qiskit, which can work with diverse $D_a$ type gates. The basic idea of XXDecomposer is to analytically compute the circuit which has the nearest Cartan coordinate to the target, within the circuit subspace $P_n$ which consists of $n$ $D_a$ gates. If the target Cartan coordinate is not in $P_n$, XXDecomposer will further consider $P_{n+1}$.

\textbf{Metrics:} The metrics considered is the \textbf{compilation time} for quantum operations, and the \textbf{decomposition cost}, which means the required number of basis two-qubit gates to implement/compile a specific quantum operation. We also consider the fidelity and the latency of the decomposed circuit. The circuit fidelity is computed with the formula $\prod_{g \in circ} (1-e_g)$, where $e_g$ is the error rate of the gate $g$. The circuit latency is normalized to the count of CX gates.

\begin{figure*}[!ht]
	\centering	 
	\begin{minipage}{0.32\textwidth}
		\centering
		\includegraphics*[height=0.70\textwidth]{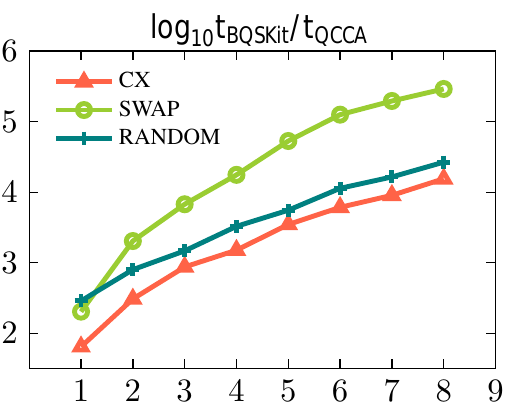} \\[-0pt]
		(a) For \frameworkname{} with $D_a$
	\end{minipage}
	\begin{minipage}{0.32\textwidth}
		\centering
		\includegraphics*[height=0.70\textwidth]{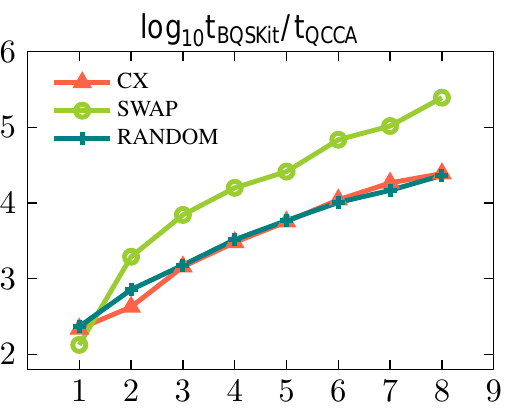} \\[-0pt]
		(b) For \frameworkname{} with $D_b$
	\end{minipage}
	\begin{minipage}{0.32\textwidth}
		\centering
		\includegraphics*[height=0.70\textwidth]{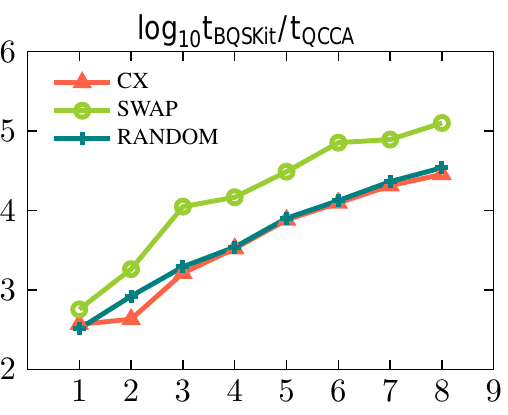} \\[-0pt]
		(c) For \frameworkname{} with $D_c$
	\end{minipage}
	\caption{Comparing the compilation time by \frameworkname{} to that by BQSKit. Both BQSKit and \frameworkname{} adopt the same 2Q basis gate in compilation. The y axis is $\log_{10} \frac{t_{BQSKit}}{t_{QCCA}}$, where $t_{BQSKit}$ and $t_{QCCA}$ are the compilation time by BQSKit and \frameworkname{}, respectively. As for the X axis, for (a), x axis value $k$ means to use $D_a$ with the rotation angle $\theta_x = \frac{\pi}{4k}$; for (b), x axis value $k$ means to use $D_b$ with the rotation angle $\theta_x = \frac{\pi}{4k}, \theta_y = \frac{\pi}{8k}$; for (c), x axis value $k$ means to use $D_c$ with the rotation angle $\theta_x = \frac{\pi}{4k}, \theta_y = \frac{\pi}{8k}, \theta_z = \frac{\pi}{16k}$.
	}\label{fig:timecomp}
\end{figure*}

\textbf{Benchmark:} For the comparison to the baseline, we use the benchmark consisting of the CX gate, the SWAP gate, and 100 random unitary operators. The CX gate is the mostly widely used two-qubit gates in various existing quantum circuits, while the SWAP gate is essential for executing quantum circuits with limited qubit connectivity. The random unitary can be used to demonstrate the average decomposition cost of compilers. Due to the randomness of the numerical compilation in the baseline, we will repeat the compilation by 100 times and compute the averaged metrics.

In quantum instruction set design task, we need application-aware context and consider four benchmark programs: Quantum Approximate Optimization Algorithm (QAOA), Quantum Fourier transform (QFT), Bernstein-Vazirani algorithm (BV) and Pauli Evolution programs (PauliEvo). 

QAOA is a promising quantum application in the NISQ (Noisy Intermediate-Scale Quantum) era, designed to solve combinatorial optimization problems. QAOA programs have also been widely used to benchmarking the performance of quantum processors~\cite{Arute2020QuantumAO}. For QAOA, Our experiments use the MAXCUT ansatz over 100 random regular graphs where we will have a ZZ interaction between two nodes with a probability 0.3. 

QFT is a key subroutine in many important quantum applications, such as Shor's algorithm~\cite{Shor1994AlgorithmsFQ} and the quantum adder~\cite{RuizPerez2014QuantumAW}. The QFT program with $n$ qubits contains $n(n-1)/2$ controlled-rotation gates, with rotation angles as small as $2\pi/2^n$. 

BV algorithm is used to learn a string encoded in a oracle function and is important for quantum cryptography applications~\cite{xie2019using}. We use BV programs with a secret string containing all 1s. 

Pauli Evolution programs like UCCSD~\cite{Fedorov2021UnitarySC} are widely used to simulate the property of physics systems. To evaluate the averaged performance, we consider the evolution of 10 random Pauli strings where each  involves at most $j$ qubits ($j$ is randomly chosen between 2 and $n$), for a $n$-qubit Pauli evolution circuit.

\textbf{Hardware modeling:} For the calibration of native gates, we follow existing schemes \cite{lin2022let,Chen2023OneGS}. 
When evaluating the design options of the quantum instruction set, we follow the established experimentally realistic hardware model for existing superconducting processors~\cite{peterson2022optimal, Chen2023OneGS}. We summarize key points as follows:
\begin{enumerate}
	\item For device Hamiltonian consisting of XX interactions between qubits~\cite{peterson2022optimal}, the latency $\tau$ of $D_a, D_b$ and $D_c$ gates is proportional to the $L1$ norm of their Cartan coordinates (i.e., $\theta_x + \theta_y + \vert \theta_z \vert$); 
	\item For device Hamiltonian consisting of both XX and YY interactions~\cite{Chen2023OneGS} between qubits, the latency $\tau$ of $D_a$ and $D_b$ gates is proportional to the $L\infty$ norm of their Cartan coordinates (i.e., the value of $\theta_x$);
	\item The error rate of $D_a, D_b$ and $D_c$ gates is proportional to their latency $\tau$~\cite{peterson2022optimal,Chen2023OneGS}, and there is an experimentally-fitted formula: $p_e = m\tau + b$, $\frac{\pi}{4}\cdot m \approx 5.76\times 10^{-3}$, $b\approx 1.909\times 10^{-3}$. 
\end{enumerate}

\begin{figure*}[!ht]
	\centering
	\begin{minipage}{0.32\textwidth}
		\centering
		\includegraphics*[width=0.95\textwidth]{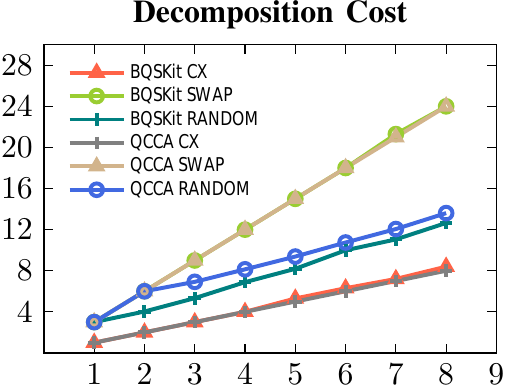} \\[-3pt]
		(a) For \frameworkname{} with $D_a$
	\end{minipage}
	\begin{minipage}{0.32\textwidth}
		\centering
		\includegraphics*[width=0.95\textwidth]{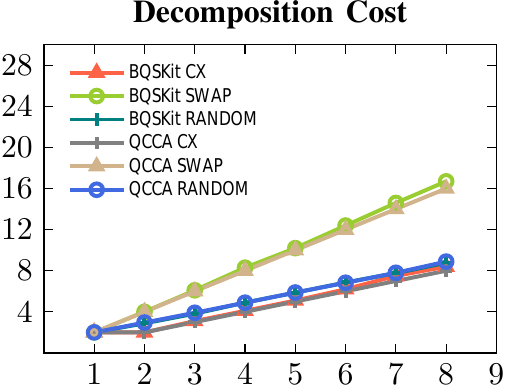} \\[-3pt]
		(b) For \frameworkname{} with $D_b$
	\end{minipage}
	\begin{minipage}{0.32\textwidth}
		\centering
		\includegraphics*[width=0.95\textwidth]{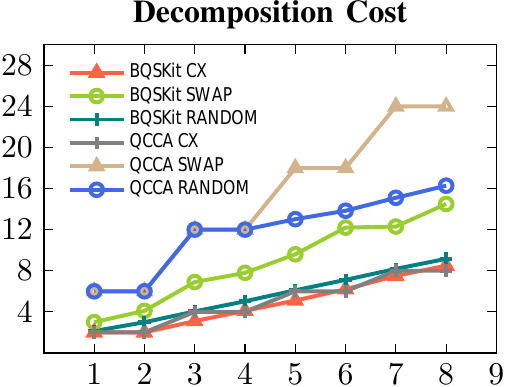} \\[-3pt]
		(c) For \frameworkname{} with $D_c$
	\end{minipage}
	\caption{Comparing the decomposition cost by \frameworkname{} to that by BQSKit. Both BQSKit and \frameworkname{} adopt the same 2Q basis gate in compilation. As for the X axis, for (a), x axis value $k$ means to use $D_a$ with the rotation angle $\theta_x = \frac{\pi}{4k}$; for (b), x axis value $k$ means to use $D_b$ with the rotation angle $\theta_x = \frac{\pi}{4k}, \theta_y = \frac{\pi}{8k}$; for (c), x axis value $k$ means to use $D_c$ with the rotation angle $\theta_x = \frac{\pi}{4k}, \theta_y = \frac{\pi}{8k}, \theta_z = \frac{\pi}{16k}$. }\label{fig:costcomp}
\end{figure*}

\subsection{Comparing to the baseline}

\begin{figure*}[t]
	\centering
	\begin{minipage}{0.45\textwidth}
		\centering
		\includegraphics*[width=\textwidth]{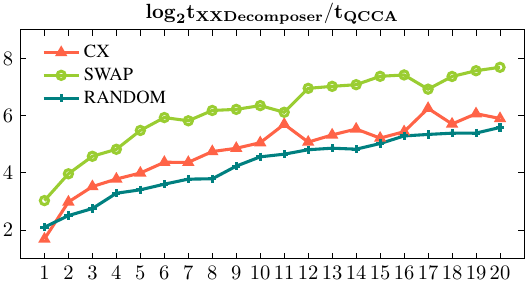} \\[-1pt]
		(a)
	\end{minipage}
	\begin{minipage}{0.45\textwidth}
		\centering
		\includegraphics*[width=\textwidth]{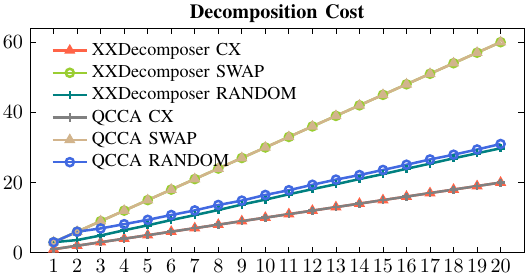} \\[-1pt]
		(b)
	\end{minipage}
	\begin{minipage}{0.45\textwidth}
		\centering
		\includegraphics*[width=\textwidth]{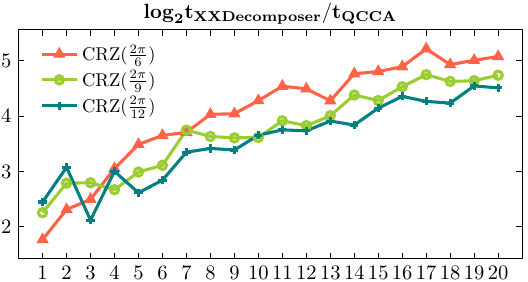} \\[-1pt]
		(c)
	\end{minipage}
	\begin{minipage}{0.45\textwidth}
		\centering
		\includegraphics*[width=\textwidth]{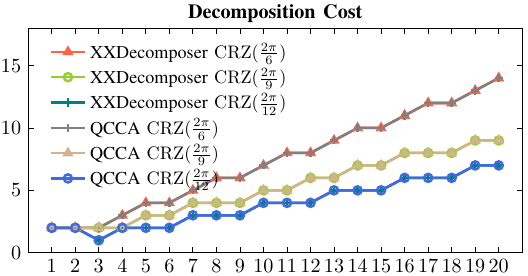} \\[-1pt]
		(d)
	\end{minipage}
	\caption{Comparing \frameworkname{} to XXDecomposer (can only works with $D_a$ type gates). The y axis is $\log_{2} t_{XXDecomposer}/t_{QCCA}$, where $t_{XXDecomposer}$ and $t_{QCCA}$ are the compilation time by XXDecomposer and \frameworkname{}, respectively. As for the X axis in (a)(b)(c)(d), x axis value $k$ means to use $D_a$ with the rotation angle $\theta_x = \frac{\pi}{4k}$.
	}\label{fig:compdaxx}
\end{figure*}

The comparison of the proposed \frameworkname{} framework to the baseline BQSKit in terms of compilation time overhead and decomposition cost are shown in 
Figure~\ref{fig:timecomp} and~\ref{fig:costcomp}, respectively. 

\noindent\textbf{Compilation time comparison between \frameworkname{} and BQSKit:} As shown in Figure \ref{fig:timecomp}, the compilation time by the BQSKit is thousands of times of that by \frameworkname, and the ratio grows bigger as the rotation angle becomes smaller. For example, considering the compilation with $D_a = e^{i\frac{\pi}{32}X_0X_1}$, the compilation time of the SWAP gate by BQSKit  is about $2.5\times 10^5$ times of the time by \frameworkname.
The time reduction by \frameworkname{} is due to its analytical nature.
For the BQSKit, as the rotation angle in $D_a, D_b$ or $D_c$ becomes smaller, a circuit with more 2Q gates is needed in BQSKit to approximate the target unitary, according to Theorem~\ref{prop:dacost}. Further, BQSKit will take $O(36n^2)$ ($n$ is 2Q gate count) time/resource to compute and store the gradient for the numerical search engine, so that to finally reach the compilation result.
Instead, the compilation time/resource by \frameworkname{} is constant across operations to be decomposed. 

It is important to have a small compilation time for 2Q quantum operators to make the compiler scalable for practical quantum computing. Many real-world quantum applications, such as Grover's algorithm~\cite{Grover1996AFQ}, often consists of millions of two-qubit quantum operators. With \frameworkname, $10^6$ two-qubit quantum operators can be processed within 20 minutes, no matter how small the rotation angle in the basis 2Q gate is. Instead, for BQSKit, even if the basis 2Q gate is just the CX gate, compiling $10^6$ two-qubit quantum operators will take about 5 days. This huge amount of compilation time will hurt the advantage of quantum processing.

\noindent\textbf{Decomposition cost comparison between \frameworkname{} and BQSKit:} As shown in Figure \ref{fig:costcomp}, the decomposition cost by \frameworkname{} is comparable to that by the BQSKit, when using $D_a$ or $D_b$. The averaged decomposition cost difference between BQSKit and \frameworkname{} is about 1 with $D_a$ and is about 0 with $D_b$, for the tested benchmark programs. This demonstrates that for $D_a$ and $D_b$, it is possible to approximate the optimal synthesis (i.e., BQSKit) with the analytical decomposition technique. The result also demonstrates that $D_b$ may be a more efficient basis gate than $D_a$ (when $\theta_x$ is the same), as $D_b$ leads to lower decomposition cost for test programs.

However, for the case with $D_c$, \frameworkname{} is not as efficient as BQSKit. 
The residual synthesis with $D_c$ in Section \ref{sect:dc} uses two $D_c$ to produce the gate of the $D_a$ type, making the decomposition cost with $D_c$ at least the multiple of 4, missing the optimization opportunity for the quantum operator whose compiled circuit indeed only requires $4k+2$ ($k\in \mathbb{Z}$) invocations of $D_c$. This demonstrates that it is necessary to take thousands of times more compilation time if we want to approximate the optimal synthesis with $D_c$. On the other hand, $D_c$ is naturally less efficient than $D_b$. No matter for BQSKit or \frameworkname, the decomposition cost with $D_c$ is on average slightly larger than that by $D_b$ when the XX rotation angle $\theta_x$ is the same, as shown in Figure \ref{fig:costcomp}. 

Overall, with $D_a$ or $D_b$, the analytical decomposition by \frameworkname{} can shorten the compilation time by thousands of times when compared to the numerical decomposition by BQSKit, at the cost of slightly increasing the decomposition cost. 

\noindent\textbf{Comparing \frameworkname{} to XXDecomposer:} The comparison of the proposed \frameworkname{} framework and the baseline XXDecomposer in terms of compilation time overhead and decomposition cost are shown in Figure~\ref{fig:compdaxx}(a)(b)(c)(d), respectively. We can see that \frameworkname{} still achieves significant speedup (up to 206.4x for the tested configurations) comparing to XXDecomposer and this speedup grows as the rotation angle in $D_a$ become smaller. 
\frameworkname{} compiles all 2Q operators globally (i.e., with a fixed compilation time), while XXDecomposer adopts a path-like searching for candidate circuits. Though XXDecomposer finds the candidate circuit analytically in each step of searching, the increased search step induced by a smaller rotation angle in $D_a$ still hurts the scalability of XXDecomposer.

As for the decomposition cost, \frameworkname{} achieves the same decomposition cost as XXDecomposer for the CX, SWAP and CRZ (controlled z-rotation) gates. For the random unitary operator, the local optimality of XXDecomposer leads to overall slightly lower compilation cost than \frameworkname{}, with a  difference about 1 averaged over all tested configurations. 
However, comparing to other 2Q quantum operators, the CX, SWAP and CRZ gates are most commonly-used among popular quantum applications, e.g., QAOA, QFT, BV and Pauli evolution circuits. Moreover, \frameworkname{} can work with native gates of $D_b$ and $D_c$ types. As a comparison, it still remains unclear how to generalize XXDecomposer to work with $D_b$ and $D_c$.

\subsection{Empirical Quantum Instruction Set Design}

In this section, we  evaluate various options of  quantum instruction set design with the proposed compiler framework. An efficient quantum instruction set can improve the fidelity of compiled circuits and reduce the latency, enhancing the quantum advantage.

\begin{figure*}[!ht]
	\centering
	\begin{minipage}{0.43\textwidth}
		\centering
		\includegraphics*[width=0.95\textwidth]{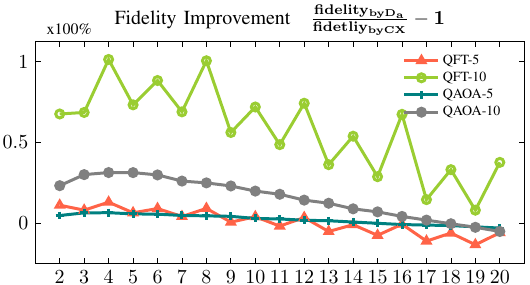} \\[-1pt]
		(a)
	\end{minipage}
	\begin{minipage}{0.43\textwidth}
		\centering
		\includegraphics*[width=0.95\textwidth]{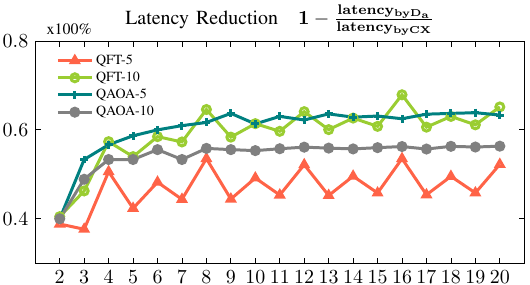} \\[-1pt]
		(b)
	\end{minipage}
	\begin{minipage}{0.43\textwidth}
		\centering
		\includegraphics*[width=0.95\textwidth]{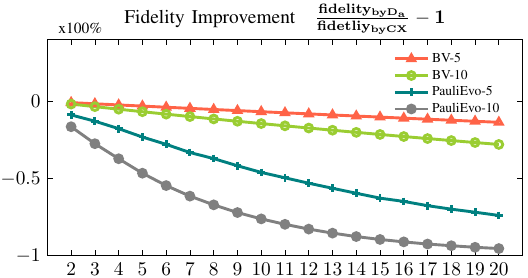} \\[-1pt]
		(c)
	\end{minipage}
	\begin{minipage}{0.43\textwidth}
		\centering
		\includegraphics*[width=0.95\textwidth]{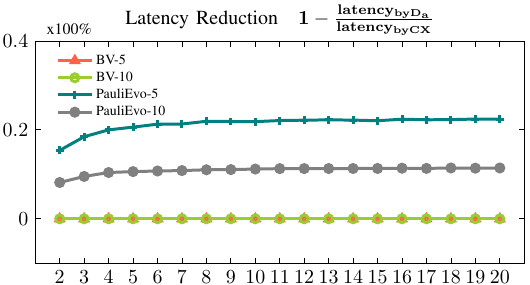} \\[-1pt]
		(d) 
	\end{minipage}
	\caption{Comparing $D_a$ to CX on compiling various quantum applications. For the X axis in (a)(b)(c)(d), x axis value $k$ means to use $D_a$ with the rotation angle $\theta_x = \frac{\pi}{4k}$. QFT-5 means a 5-qubit QFT circuit. Other legend labels are named in the same manner.}\label{fig:smallbench}
\end{figure*}

\textbf{(Q1) Why using basis 2Q gate with a small rotation angle?}

As shown in Figure~\ref{fig:smallbench}(a), using $D_a$ with some small rotation angles (e.g., $\theta_x = \frac{\pi}{16}$) can improve the fidelity of QFT and QAOA circuits. This is because there are many small rotation angles in the CRZ gates of the QFT circuit and in the ZZ interaction gate of QAOA circuits. We may only use two invocations of $D_a$ to implement those gates. Even with the CX gate, we still need two invocations of CX  to compile those gates. 
According to the experimentally-realistic error model \cite{peterson2022optimal,Chen2023OneGS}, the error rate of the CX gate is far larger than that of the $D_a$ gate. Thus, for circuits consisting of many 2Q operators that has a small Cartan coordinate, using $D_a$ gate with a small rotation angle will provide considerable improvement for the fidelity of the compiled circuit. We also points out that, the fidelity improvement is bounded by the amount rather than the percentage of 2Q operators with a small Cartan coordinate, due to the fact that the circuit fidelity is computed as the product of individual gate fidelity. Thus, the fidelity improvement becomes larger when the QFT and QAOA program involve more qubits, as shown in Figure~\ref{fig:smallbench}(a). 

The latency reduction induced by using $D_a$ with a small $\theta_x$ (see Figure \ref{fig:smallbench}(b)) is due to the same reason. However, this latency reduction is bounded by the percentage of 2Q operators with a small Cartan coordinate, existing in QFT and QAOA circuits. This means the latency reduction induced by using a smaller rotation angle in $D_a$ will gradually converge to some fixed value, as supported by the tendency shown in Figure \ref{fig:smallbench}(b). On the other hand, using $D_a$ with a too small rotation angle will hurt the circuit fidelity due to the increased 2Q gate count, as shown in Figure \ref{fig:smallbench}(a). 
Thus, to find the $D_a$ which maximally boosts the circuit performance, we need to make a good balance between the fidelity improvement and the latency reduction. 

Compared to other state-of-the-art tools, \frameworkname{} can help finding the optimal rotation angle for $D_a$, $D_b$ or $D_c$ much more efficiently due to its extremely low time and resource overhead. \frameworkname{} can decompose a quantum program with $10^3$ different 2Q operators and thus getting its fidelity and latency data within one second, making \frameworkname{} useful not only for the quantum instruction set design in the near term, but also for future large-scale quantum applications.

\textbf{(Q2) Will having mixed 2Q basis gates in the instruction set be more efficient?}

By comparing Figure \ref{fig:compdaxx}(c) to (a), we can see that, for different benchmark programs, the rotation angle of $D_a$ that leads to the best circuit fidelity is different.  For the BV and Pauli evolution benchmark, most 2Q operators are just the CX gate, and the fidelity of CX indirectly implemented with $D_a$ ($\theta_x < \frac{\pi}{4}$) is lower than the directly implemented CX on hardware. In this case, the fidelity improvement will decrease as the circuit size increases or as $\theta_x$ decreases.

Therefore, having $D_a$ with different rotation angles in the instruction set is important for always obtaining circuit fidelity improvement over various quantum applications. For the tested eight benchmark programs, we prefer the mixed instruction set $\{CX, D_a(\frac{\pi}{16})\}$ if putting more emphasis on the fidelity improvement, or $\{CX, D_a(\frac{\pi}{32})\}$ if putting more emphasis on the latency reduction. This design also works for the $D_b$ gate since for those programs, the required number of invocations is the same for $D_a$ and $D_b$ if their $\theta_x$ is the same.

In the general case, hardware engineers can use \frameworkname{} to test more benchmark programs, and select appropriate rotation angles for 2Q basis gates in an application-aware way demonstrated above. Since \frameworkname{} naturally supports hybrid compilation, e.g., enabling simultaneous usage of $D_a$ gates with varied rotation angles, hardware engineers do not need to worry about the porting of quantum circuits. 
Researchers~\cite{lin2022let,Chen2023OneGS} have proposed various calibration methods for 2Q basis gates. For a mixed instruction set $G$ with $m$ different basis gates, its calibration overhead is proportional to $m$. With \frameworkname{}, it is easy to take the calibration cost into the mixed instruction set design, only needing to add a term related to the calibration cost in the objective function that measures the efficiency of the instruction set $G$.

\textbf{(Q3) Which type of gates is most efficient? $D_a$, $D_b$, or $D_c$?}

We rule out $D_c$ for the following reasons. Firstly, the decomposition cost by $D_b$ is comparable to that by $D_c$ if their $\theta_x$ and $\theta_y$ are the same, as shown in Figure~\ref{fig:costcomp}(b)(c). Secondly, $D_c$ gates are hard to prepare in existing quantum platforms \cite{smith2016practical,qiskit2024,Chen2023OneGS}, whose device Hamiltonian either only consists of XX interactions or XX+YY interactions. In such cases, we must use $D_a$ or $D_b$ to implement $D_c$ in an indirect manner. Last but not least, as shown in Figure~\ref{fig:costcomp}(c), when using $D_c$, it is hard to approach the optimal compilation with analytical schemes. This fact will cause great difficulty in designing a near-optimal and scalable compiler based on $D_c$.

Whether $D_a$ or $D_b$ is more efficient depends on the hardware platform. For the quantum processor based on the XX coupling between qubits, the instruction combination $G_1 = \{D_a(\theta_x), D_a(\theta_y)\}$ is more efficient than $G_2 = \{D_b(\theta_x, \theta_y)\}$ since the gate $D_b(\theta_x, \theta_y)$ is implemented by $D_a(\theta_x)S_{0,1}D_a(\theta_y)S^\dagger_{0,1}$ in such a platform. Any fidelity improvement and latency reduction achieved by $G_2$ can thus also be achieved by $G_1$. On the other hand, $G_1$ can achieve efficient decomposition that cannot be achieved by $G_2$. For example, $e^{i(\theta_x+\theta_y)X_0X_1}$ can be achieved by $D_a(\theta_x)D_a(\theta_y)$ and $D_b Y_1(\beta) D_b^\dagger$. It is obvious that the former circuit has a smaller latency and a higher fidelity.

For the hardware platform based on the XX+YY coupling between qubits, we argue $D_b$ is more efficient. In such a platform, for a fixed $\theta_x$, $D_b(\theta_x, \theta_y)$ has a similar latency and fidelity as $D_a(\theta_x) = D_b(\theta_x, 0)$, according to the established hardware model \cite{Chen2023OneGS}. While the performance of such $D_b$ and $D_a$ are almost identical for benchmark programs whose 2Q operators are locally equivalent to XX interactions, the extra $\theta_y$ in $D_b$ can largely reduce the decomposition cost of the SWAP gate (see Figure~\ref{fig:costcomp}(a)(b)), which appears frequently (up to 4x the count of other 2Q operators \cite{Iten2015QuantumCF}) in qubit-connectivity-limited quantum processors. We can further find the optimal $\theta_y$ for the fixed $\theta_x$ by examining the decomposition cost with \frameworkname. For example, for $\theta_x = \frac{\pi}{4}$, $\theta_y = \frac{\pi}{8}$ can achieve the best decomposition cost for both random unitary operators and the SWAP gate; for $\theta_x = \frac{\pi}{8}$, $\theta_y = \frac{\pi}{8}$ is the best choice.

The  discussion  above invalidates the option of having both $D_a$ and $D_b$ in the quantum instruction set. This is because the most efficient type of gates are determined by the type of qubit coupling. We may have both $D_a$ and $D_b$ in the instruction set only when the device Hamiltonian has mixed types of coupling between different qubit pairs. No matter in which case, \frameworkname's inherent capability of supporting compilation with hybrid native gate combination always makes it easy to port quantum circuits to different hardware platforms and instruction sets.

\section{Related Work}

Works that directly relate to our paper target quantum compilation and quantum instruction design. 
For quantum compilation works, most existing compilers \cite{khaneja2000cartan, shende2005synthesis, shende2003minimal, dawson2005solovay, barenco1995elementary, Bergholm2004QuantumCW, de2016block, Vatan2003OptimalQC, Vidal2003UniversalQC, Aho2003CompilingQC, Cybenko2001ReducingQC, Vartiainen2003EfficientDO, Mttnen2004QuantumCF,Tucci1998ARQ,smith2016practical,Sivarajah2020tketAR,Smith2020AnOI,Amy2019staqAFQ,Khammassi2020OpenQLA,cirq,qiskit2024} assume using a few standard two-qubit gates like the CX/CZ and iSWAP gate in the compilation. Those compilers cannot work with an arbitrarily entangling two-qubit gate, let alone evaluate the pros and cons of different two-qubit basis gates.

On the other hand, researchers have tried to extend the border of the quantum instruction set architecture and design new quantum instructions, e.g., the continuous fsim gate family \cite{Nguyen2022ProgrammableHI,foxen2020demonstrating}, XY gate family \cite{Abrams2020ImplementationOX}, and CRISC gate family \cite{Chen2023OneGS}. However, those works only provide analytical decomposition schemes for a few fixed two-qubit gates,  like the $\sqrt{iSWAP}$ gate\cite{Prez2021ErrorDivisibleTG}. They do not support analytical decomposition with arbitrary 2Q entangling gates. Lao et al.\cite{lao2021designing} and Lin et al.\cite{lin2022let} propose to use the numerical compilation method \cite{BQSKIT} so as to support all kinds of two-qubit basis gates. However, numerical methods often suffer from heavy time complexity, hindering wider exploration of the quantum instruction set design space.
Peterson et al.\cite{peterson2022optimal} proposes a locally analytical compiler for all gates in the XX gate family. However, their tool cannot work with 2Q basis gates of the $D_b$ and $D_c$ type. 

\section{Conclusion and Future Work}

This paper proposes an analytical framework \frameworkname{} to analyze the Cartan coordinate of two-qubit operators and proposes a series of constructions to convert between 2Q operators based on the feature of their Cartan coordinates. This work further proposes a compiler that significantly reduces the conversion cost between two-qubit operators and demonstrates comparable overall decomposition results with the numerically optimal compiler. Compared to the state-of-the-art, \frameworkname{} reduces the time overhead of quantum instruction set design by several orders of magnitude. With \frameworkname, 
we can also suppress circuit errors by including noise-induced gate invariant into the instruction set for compilation.

While this work establishes comprehensive analytical compilation with diverse native 2Q gates, there is still much space left for improvements.

\textbf{Algorithm construction with diverse native gates:} This work mainly discusses how to use various native 2Q gates to construct other 2Q operators. Another interesting direction is to explore how to use $D_a$ (or $D_b$ or $D_c$) to directly synthesize n-qubit quantum operations and algorithms, such as the Pauli string evolution, without sticking to the CX-based synthesis. This may help us gain benefits with the native 2Q gate that has a small Cartan coordinate, even if the benchmark program is full of CX-like gates.

\textbf{Extending to higher dimensions:} This work mainly works with diverse 2Q native gates. It is possible to have higher dimensional native gates, e.g., CCZ gate,  on quantum hardware platforms like those based on neutral atoms \cite{Evered2023HighfidelityPE}. It is interesting to explore whether there is analytical conversion schemes between 3Q gates, or 2Q and 3Q gates. Moreover, to what upper bound dimension, analytical conversion between native gates exists. This exploration can help establish more expressive intermediate representation in compilation, enhancing the portability of quantum circuits.

\section*{Appendix}
\label{appendix}

We illustrate the proof of Theorem~\ref{prop:dacost} as follows. The proof is based on the KAK decomposition, of which the details can be founded in \cite{tucci2005introduction}. 
For a given $U \in SU(4)$, consider 
$$M^\dagger U M = U_r + iU_i = Q_L (D_r + i D_i)Q_R^T,$$
where $Q_L, Q_R$ are orthogonal matrices, $D_r, D_i$ are diagonal matrices,
$M = \frac{1}{\sqrt{2}}$
\begingroup \tiny
$\begin{pmatrix}
	1 & 0 & 0 & i \\
	0 & i & 1 & 0 \\
	0 & i & -1 & 0 \\
	1 & 0 & 0 & -i
\end{pmatrix}
$
\endgroup
$,(D_r + i D_i) = exp(i\, \text{diag}(\eta_0,\eta_1,\eta_2,\eta_3)),
$
with $\eta_i$ determined by the global phase $g \in \{0, \frac{\pi}{2}\}$ and the Cartan coordinate ($k_x, k_y, k_z$), $\frac{\pi}{4} \ge k_x \ge k_y \ge \vert k_z \vert \ge 0$:
\begingroup \small
\begin{equation}
	\begin{pmatrix}
		\eta_0 \\
		\eta_1 \\
		\eta_2 \\
		\eta_3 \\
	\end{pmatrix} = \begin{pmatrix}
		+1 & +1 & -1 & +1 \\
		+1 & +1 & +1 & -1 \\
		+1 & -1 & -1 & -1 \\
		+1 & -1 & +1 & +1
	\end{pmatrix} \begin{pmatrix}
		g \\
		k_x \\
		k_y \\
		k_z \\
	\end{pmatrix}
\end{equation}
\endgroup

Without ambiguity, here $\sigma_i(U)$ refers to the $i$-th maximal singular value of $U$.
\begin{prop}\label{prop:sigma} If $g=0$, then the singular value of $U_i$ is:
	\begin{equation}\label{equ:uisingular}\small
		\begin{pmatrix}
			\sigma_1(U_i) \\
			\sigma_2(U_i) \\
			\sigma_3(U_i) \\
			\sigma_4(U_i) 
		\end{pmatrix} = \begin{pmatrix}
			\sin (k_x + k_y + k^+_z ) \\
			\sin (k_x + k_y - k^+_z) \\ 
			\sin (k_x - k_y + k^+_z) \\
			\sin \vert -k_x + k_y + k^+_z \vert
		\end{pmatrix},
	\end{equation}
	where $k_z^+ = \vert k_z \vert$. If $g=\frac{\pi}{2}$, we only need to replace $U_i$  with $U_r$.
\end{prop}
\begin{proof}
	Firstly, $\vert D_i \vert$ contains all singular values of $U_i$ regardless the ordering. For example, $\vert\sin \eta_2 \vert$ is just $\sin (k_x + k_y + k_z )$ if $k_z \ge 0$. As for the ordering of those singular values:
	
	\noindent (1) $ \sin (k_x + k_y + k^+_z ) \ge \sin (k_x + k_y - k^+_z) \Longleftrightarrow k_x + k_y + k^+_z \le \pi - (k_x + k_y - k^+_z) \Longleftrightarrow k_x + k_y \le \frac{\pi}{2}$;\smallskip
	
	\noindent (2) $\sin (k_x + k_y - k^+_z ) \ge \sin (k_x - k_y + k^+_z) \Longleftrightarrow k_y \ge k_z^+$;\smallskip
	
	\noindent (3) $ \sin (k_x - k_y + k^+_z) \ge \sin \vert -k_x + k_y + k^+_z \vert \Longleftrightarrow k_x - k_y + k^+_z \ge \vert -k_x + k_y + k^+_z \vert$ \\
	If $k_y + k^+_z \ge k_x$, the inequality is equivalent to $k_x \ge k_y$; \\
	If $k_y + k^+_z \le k_x$, the inequality is equivalent to $k_z^+ \ge 0$; 
	
	Therefore, the singular values of $U_i$ is just as shown in the proposition. The case for $U_r$ can be computed similarly.
\end{proof}

We define $k_t(U) = (k_x + k_y + k_z^+)(U)$, $R(\beta) = \min (\beta, \pi-\beta)$, $\theta=\theta_x$ (of $D_a$). We assume $0< \theta < \frac{\pi}{4}$.
\begin{prop}[Triangle inequality]\label{prop:datrianglecos}
	For $U \in SU(4)$, 
	if $R(k_t(U)) \le \frac{\pi}{4}$,  $\vert \cos k_t(UD_a) \vert \le \cos (R(k_t(U))-\theta)$.
\end{prop}
\noindent\textit{Proof}.
Let $U_{1} = U D_a$, then
	$U_{1r} + i U_{1i} = M^\dagger U_{1} M = M^\dagger  U M e^{i\theta Z_0} \nonumber \\
	= (\cos\theta U_{r} - \sin\theta U_{i}Z_0) + i (\cos\theta U_{i} + \sin\theta U_{r}Z_0)$.
Since $Z_0 \in SO(4)$, the singular values of $U_{i}Z_0$ should be as the same as $U_{i}$. Then according to Weyl's Inequality \cite{franklin2012matrix},
\begingroup \small
\begin{align}
	\sigma_4(U_{1r}) \le \cos\theta \sigma_4(U_{r}) + \sin\theta \sigma_1(U_{i}),\label{equ:u1r} \\
	\sigma_4(U_{1i}) \le \sin\theta \sigma_4(U_{i}) + \cos\theta \sigma_1(U_{r}).\label{equ:u1i} \\
	\sigma_4(U_{1r}) \le \cos\theta \sigma_1(U_{r}) + \sin\theta \sigma_4(U_{i}),\label{equ:u1r1} \\
	\sigma_4(U_{1i}) \le \sin\theta \sigma_1(U_{r}) + \cos\theta \sigma_4(U_{i}),\label{equ:u1i1}
\end{align}
\endgroup
Now, we do case by case study regarding the global phase $g_{U_1}$ and $g_U$ of $U_1$ and $U$, respectively. 

(i) If $g_{U_1}=0$ and $g_{U}=0$, then by Proposition~\ref{prop:sigma}, Equation \ref{equ:u1r} is equivalent to
	$\vert\cos (k_t(U_1))\vert \le \cos\theta \vert\cos k_t(U) \vert + \sin\theta \sin (k_t(U)), 
	 = \cos (R(k_t(U))-\theta)$. 
This completes the proof.

(ii) If $g_{U_1}=\frac{\pi}{2}$ and $g_{U}=0$, then by Proposition~\ref{prop:sigma}, Equation \ref{equ:u1i} becomes
	$\vert\cos (k_t(U_1))\vert \le \cos\theta sin k_t(U) + \sin\theta \vert\cos (k_t(U))\vert \\
	= \sin(R(k_t(U))+\theta) 
	\le  \cos(R(k_t(U))-\theta)$.
This completes the proof. The second inequality is due to $R(k_t(U)) \le \frac{\pi}{4}$. 

Other cases, such as $(g_{U_1}, g_{U})=(0, \frac{\pi}{2})$ and $(g_{U_1}, g_{U})=(\frac{\pi}{2}, \frac{\pi}{2})$ can be discussed similarly.

Now we start to prove Theorem \ref{prop:dacost}.
\begin{proof}
	The theorem can be proved by contradiction:
	
	Assume $sin\, k_t(U) > \sin N\theta$. This means $R(k_t(U)) > N\theta$. Now consider $U_1 = U (a_1\otimes b_1) D_a^\dagger (c_1 \otimes d_1)$, where $a_1, b_1, c_1, d_1 \in SU(2)$. Note that single-qubit gates do not affect the Cartan coordinate. Then we discuss $R(k_t(U_1))$ case by case:
	
	(i) if $R(k_t(U)) \le \frac{\pi}{4}$, $\vert \cos k_t(U_1) \vert \le \cos (R(k_t(U)) - \theta)$, according to Proposition~\ref{prop:datrianglecos}.
	This indicates $R(k_t(U_1)) \ge R(k_t(U)) - \theta$. 
	
	(ii) If $R(k_t(U)) > \frac{\pi}{4}$ and $g_{U_1} = g_U$, then we still have $R(k_t(U_1)) \ge R(k_t(U)) - \theta$, following the proof of Proposition \ref{prop:datrianglecos};
	
	(iii) If $R(k_t(U)) > \frac{\pi}{4}$ and $g_{U_1} \ne g_U$, following the proof of Proposition \ref{prop:datrianglecos} (now using Equation \ref{equ:u1i1} for the case (ii) in the proof of Proposition \ref{prop:datrianglecos}), we have 
		$\vert\cos (k_t(U_1))\vert \le \cos\theta \sin \vert -k_x + k_y + k_z^+ \vert 
		+ \sin\theta \cos \vert -k_x + k_y + k_z^+ \vert 
		= \sin \vert -k_x + k_y + k_z^+ \vert + \theta 
		\le \sin (\frac{\pi}{4}+\theta) = \cos(\frac{\pi}{4}-\theta)$.
	This leads to $R(k_t(U_1)) \ge \frac{\pi}{4} - \theta$.
	
	Overall, we have 
	$R(k_t(U_1)) \ge \min(\frac{\pi}{4}, R(k_t(U))) - \theta > (N-1)\theta.$
	Then for $U_{N} = U_{N-1} (a_{N-1}\otimes b_{N-1}) D_a^\dagger (c_{N-1} \otimes d_{N-1})$, with recursion, 
	$R(k_t(U_N))  >  (N-N)\theta = 0$.
	
	However, by adjusting single-qubit gates $\{a_i, b_i, c_i, d_i\}$, the generated set of $U_N$ should contain the identity gate. In this case, $0 =  R(k_t(I)) = R(k_t(U_N)) > 0$, causing the contradiction. Thus, $\sin k_t(U) \le \sin N\theta$.
\end{proof}

\bibliographystyle{unsrt}
\bibliography{sample-base}

\end{document}